\RequirePackage{fix-cm}
\documentclass[smallextended]{svjour3}       
\smartqed  
\journalname{Journal of Elasticity}
\usepackage{amsmath,bm,color,xcolor}
\usepackage[pdftex]{graphicx}
\usepackage{amsfonts}
\usepackage{enumerate}
\usepackage{amssymb}
\usepackage{mathrsfs}
\usepackage{cancel}
\usepackage{graphicx}
\usepackage{color}
\usepackage{url}

\newcommand{\om}{\Omega}

\newcommand{\R}{\mathbb{R}}

\newcommand{\n}{{\mathbf n}}
\newcommand{\bv}{{\mathbf b}}
\newcommand{\muv}{{\bm \mu}}
\newcommand{\x}{{\mathbf x}}
\newcommand{\kv}{{\mathbf k}}
\newcommand{\rv}{{\mathbf r}}
\newcommand{\m}{{\mathbf m}}
\newcommand{\y}{{\mathbf y}}

\newcommand{\vv}{{\mathbf v}}
\newcommand{\cv}{{\mathbf c}}
\newcommand{\dv}{{\mathbf d}}
\newcommand{\p}{{\mathbf p}}
\newcommand{\q}{{\mathbf q}}
\newcommand{\Q}{{\mathbf Q}}

\newcommand{\Rv}{{\mathbf R}}
\newcommand{\A}{{\mathbf A}}
\newcommand{\av}{{\mathbf a}}
\newcommand{\B}{{\mathbf B}}
\newcommand{\F}{{\mathbf F}}
\newcommand{\G}{{\mathbf G}}
\newcommand{\Pv}{{\mathbf P}}
\newcommand{\nuv}{{\bm \nu}}
\newcommand{\omv}{{\bm \omega}}

\newcommand{\e}{{\mathbf e}}

\newcommand{\C}{\mathbf{C}}
\newcommand{\mL}{{\mathcal L}}

\newcommand{\M}{{\mathbf M}}
\newcommand{\bzero}{{\mathbf 0}}

\newcommand{\tr}{\rm tr\,}

\newcommand{\SO}{{\rm SO}}
\newcommand{\Othree}{{\rm O} (3)}

\newcommand{\SOthree}{\SO (3)}

\newcommand{\be}{\begin{eqnarray}}
\newcommand{\ee}{\end{eqnarray}}

\renewcommand{\leq}{\leqslant}
\renewcommand{\geq}{\geqslant}

\newcommand{\half}{\frac{1}{2}}
\newcommand{\1}{{\mathbf 1}}

\newcommand{\Zed}{{\mathbb Z}}

\begin{document}

\title{Slip and twinning  in Bravais lattices
}


\author{John M. Ball        
}


\institute{John M. Ball \at
             Heriot-Watt University and Maxwell Institute for Mathematical Sciences, Edinburgh. \email{jb101@hw.ac.uk}}
             

\date{}

\maketitle\vspace{-1in}
\centerline{\it In memoriam Jerry Ericksen, from whom I learnt so much.}\vspace{.2in}

\begin{abstract}
A unified treatment of slip and twinning in Bravais lattices is given, focussing on the case of cubic symmetry, and using the Ericksen energy well formulation, so that interfaces correspond to rank-one connections between the infinitely many crystallographically equivalent energy wells. Twins are defined to be such rank-one connections involving a nontrivial reflection of the lattice across some plane. The slips and twins minimizing shear magnitude for cubic lattices are rigorously calculated, and the conjugates of these and other slips analyzed. It is observed that all rank-one connections between the energy wells for the dual of a Bravais lattice can be obtained explicitly from those for the original lattice, so that in particular the rank-one connections for fcc can be obtained explicitly from those for bcc.
\end{abstract}

\section{Introduction}
\label{intro}
This paper gives a unified treatment of slip and twinning in Bravais lattices, focussing on the case of cubic symmetry, and using the Ericksen energy well formulation, so that interfaces correspond to rank-one connections between the infinitely many crystallographically equivalent energy wells. Inevitably there is a considerable overlap with other such treatments, most notably with that of Pitteri \& Zanzotto \cite{pitterizanzotto03} (especially Chapter 8). 

A main contribution is that we calculate rigorously the slips and Type 1/Type 2 twins $\F=\1+\av\otimes\n$ that minimize the shear magnitude $|\av|$ for cubic lattices, these corresponding to well-known statements in the materials science literature. This reduces to the minimization of products $G(\p)H(\q)$ of quadratic forms,   where $\p,\q\in\Zed^3$ are such that $\p\cdot\q=0$ for slip and $\p\cdot\q=2$ for Type 1/Type 2 twins.  The principle of  minimizing  $|\av|$ as a means of selecting preferred slips and twins has been proposed for slips by Chalmers \& Martius \cite{chalmers1952slip} and for twins by Kiho \cite{kiho1954,kiho1958} and Jaswon  \& Dove \cite{jaswondove1957}.  In fact the Type 1/Type 2 twins minimize the shear magnitude among all rank-one connections between the energy wells. We also calculate the conjugates to the minimizing slips, together with certain other slips for bcc, showing that they are all twins. Further, we determine  all the slips for cubic lattices whose conjugates are also slips.

An apparently new observation is that all rank-one connections for the dual of a Bravais lattice can be obtained explicitly from those for the original lattice. In particular the rank-one connections for fcc lattices can be obtained explicitly in terms of those for bcc lattices (and vice versa). 
 
The plan of the paper is as follows. In Section \ref{bravais} we review standard material on Bravais lattices. Here and throughout the paper we use a matrix formulation, so that the lattice vectors $\bv_1,\bv_2,\bv_3$ are written as the columns of a matrix $\B$; at least for the author this makes calculations easier to follow.  We define lattice planes and the dual lattice, largely following Ashcroft \& Mermin \cite{ashcroftmermin}.  Then in Section \ref{ericksen} we recall the Ericksen energy-well picture, whereby the  macroscopic free-energy density inherits an infinite family of crystallographically equivalent energy wells from those of the lattice free-energy via the Cauchy-Born rule. 

In Section \ref{interfaces} we review standard results for the existence
 of rank-one connections as well as establishing (Theorem \ref{bccfcc}) the correspondence between rank-one connections for Bravais lattices and their duals. In Section \ref{slip} we give a characterization of slip systems in terms of rank-one connections not involving lattice rotation (Theorem \ref{slipchar}), and   determine the slips that minimize the shear magnitude (Theorem \ref{slipthm}) for the simple cubic, bcc and fcc lattices. In Section \ref{twins} we give a definition of twins in terms of interfaces separating nontrivially reflected lattices, and determine (Theorem \ref{twinthm}) the Type 1/Type 2 twins minimizing the shear magnitude for the simple cubic, bcc and fcc lattices.  In Section \ref{gen} we address the problem of determining general rank-one connections of minimum shear magnitude, showing in particular (Corollary \ref{cor1}) that the twins in Theorem \ref{twinthm} minimize the shear amplitude among all rank-one connections. Finally in Section \ref{conjslip} we determine (Theorem \ref{slipslip}) all slips for cubic lattices whose conjugates are also slips, and show (Theorems \ref{bccslipconj}, \ref{bccslipconj1}) that the conjugates of all the specific slips discussed in Section \ref{slip} are twins of various types. 
 
This paper focusses on cubic Bravais lattices, but extensions of the results to noncubic lattices, multilattices and crystals undergoing martensitic phase transformations would be valuable.

\section{Bravais lattices}
\label{bravais}
\subsection{Definitions and basic properties}
Denote by $\e_i$  the unit vector in the $i^{th}\rm$ coordinate direction of $\R^3$. A {\it Bravais lattice} is an infinite lattice of points in $\R^3$ generated by  linear combinations with integer coefficients of three linearly independent basis vectors $\bv_1, \bv_2,\bv_3$. Representing these basis vectors with respect to the orthonormal basis $\{\e_i\}$, and letting $\B=(\bv_1,\bv_2,\bv_3)$ be the matrix with columns $\bv_i$,  so that $B_{ij}=\bv_j\cdot \e_i$, we write the corresponding Bravais lattice as
\begin{eqnarray*}{\mathcal L}(\B)&=&\{m_1\bv_1+m_2\bv_2+m_3\bv_3: m_i\in {\mathbb Z}\}\\
&=&\{\B\m:\m\in{\mathbb Z}^3\},
\end{eqnarray*}
where ${\mathbb Z}$ denotes the set of integers. 
All lattice sites $\p\in{\mathcal L}(\B)$ are equivalent,  i.e. $$\p+{\mathcal L}(\B)={\mathcal L}(\B).$$
We denote by $\R^{3\times 3}$ the space of real $3\times 3$ matrices, by $\Zed^{3\times 3}$ the set of $3\times 3$ matrices with integer entries, and
\begin{align*}
GL(3,\R)&=\{\F\in \R^{3\times 3}: \det\F\neq 0\},\\
GL^+(3,\R)&=\{\F\in \R^{3\times 3}:  \det\F > 0\},\\
GL(3,\mathbb Z)&=\{{\bm\mu\in\Zed^{3\times 3}}:  \det{\bm\mu}=\pm1\},\\
GL^+(3,\mathbb Z)&=\{{\bm\mu\in\Zed^{3\times 3}}:  \det{\bm\mu}=1\}=SL(3,\mathbb Z),\\
{\rm O}(3)&=\{\Q\in \R^{3\times 3}:\Q^T\Q=\1\},\\
{\rm SO}(3)&=\{\Q\in{\rm O}(3):\det\Q=1\}.
\end{align*}
For $\F\in \R^{3\times 3}$ we denote the Euclidean norm of $\F$ by $|\F|:=(\tr \F^T\F)^\half$.

The following standard theorem (see, for example, \cite{Ericksen77},\cite[Proposition 3.1]{pitterizanzotto03}) characterizes the sets of basis vectors that are equivalent in the sense that they generate the same Bravais lattice.
\begin{theorem}
\label{equiv}
${\mathcal L}(\B)={\mathcal L}(\mathbf C)$ if and only if
$$\mathbf C=\B {\bm\mu}, \mbox{ for some {\boldmath $\mu$}}\in GL(3,\mathbb Z).$$
\end{theorem}
\begin{proof}
Let $\B=(\bv_1,\bv_2,\bv_3),\;\mathbf C=(\cv_1,\cv_2,\cv_3)$. If ${\mathcal L}(\B)={\mathcal L}({\bf C})$ then $\bv_i=\mu_{ji}{\bf c}_j$ for some $\bm \mu=(\mu_{ij})\in {\mathbb Z}^{3\times 3}$, so that $\B={\bf C}{\bm \mu}$. Similarly ${\bf C}=\B{\bm \mu}'$ for some ${\bm \mu}'\in{\mathbb Z}^{3\times 3}$.
So ${\bm\mu}'={\bm\mu}^{-1}$ and $\bm\mu\in GL(3,{\mathbb Z})$.

Conversely, if $\B={\bf C}{\bm\mu}$ then $\bv_i={\mu}_{ji}{\bf c}_j$ and so ${\mathcal L}(\B)\subset{\mathcal L}(\bf C)$. Similarly  ${\mathcal L}(\bf C)\subset{\mathcal L}(\B)$. \qed
\end{proof}
\begin{corollary}
\label{cor2}
If ${\bf F}\in GL(3,\mathbb R)$, then ${\mathcal L}({\bf F}\B)={\mathcal L}(\B)$ if and only if
$${\bf F}=\B{\bm\mu}\B^{-1} \mbox { for some }{\bm \mu}\in GL(3,{\mathbb Z}).$$
\end{corollary}

The {\it point group} $P(\B)$ is
\begin{eqnarray}\nonumber
P(\B)&=&\{\Q\in \Othree:\mathcal L(\Q\B)=\mathcal L(\B)\}\\&=&\{\Q\in \Othree:\Q=\B\muv\B^{-1}\text{ for some }\muv\in GL(3,\mathbb Z)\}.\label{pointgroup}
\end{eqnarray}
\subsection{Lattice planes and the dual lattice}
\label{lplanes}
A {\it lattice plane} is a plane $\Pi(\n)=\{\x\in\R^3:\x\cdot\n=k\}$ with unit normal $\n$ such that $\Pi(\n)\cap{\mathcal L}(\B)$ contains 3 non-collinear points. Equivalently, $\Pi(\n)\cap{\mathcal L}(\B)$ is a translate of a 2D Bravais lattice of the form
$$\{r_1\m_1+r_2\m_2: r_1,r_2\in {\mathbb Z}\},$$
where $\m_1,\m_2\in {\mathcal L}(\B)$ and $\m_1,\m_2$ are linearly independent. (Taking without loss of generality $k=0$ this can be proved by first choosing $\m_1$ to be a nonzero vector in $\Pi(\n)\cap{\mathcal L}(\B)$ of minimum length, and then $\m_2$ a nonzero vector  in $\Pi(\n)\cap{\mathcal L}(\B)$ not parallel to $\m_1$ and of minimum length.)

Then, following Ashcroft \& Mermin \cite{ashcroftmermin}, there exists $\p\in{\mathcal L}(\B)$ with $\p\cdot\n>0$ such that 
$${\mathcal L}(\B)=\bigcup_{r\in{\mathbb Z}}\left(\left({\mathcal L}\left(\B\right)\cap\Pi\left(\n\right)\right)+r\p\right),$$
so that ${\mathcal L}(\B)$ is the union of its intersection with a family of planes parallel to $\Pi(\n)$ with interplane spacing 
$d=\p\cdot\n$.

The {\it dual} (or {\it reciprocal}) {\it lattice} of ${\mathcal L}(\B)$ is the set
$$\{\kv\in\R^3:\kv\cdot\p\in{\mathbb Z} \text{ for all }\p\in{\mathcal L}(\B)\}.$$
The vector $\kv$ belongs to the dual lattice if and only if $\kv\cdot\B\e_i\in {\mathbb Z}$ for each $i$, which holds if and only if $\kv=\sum_{i=1}^3r_i\B^{-T}\e_i$ for $\rv\in{\mathbb Z}^3$. Hence the dual lattice is the Bravais lattice $\mathcal L(\B^{-T})$.

\begin{theorem}[see {\cite[p90]{ashcroftmermin}}]
\label{lplane} $\Pi(\n)=\{\x:\x\cdot \n=0\}$ is a lattice plane if and only if there exists $\kv\in {\mathcal L}(\B^{-T})\setminus \{0\}$ with $\kv$ parallel to $\n$, and then  $d^{-1}$ is the minimum length of such a vector $\kv$.
\end{theorem}
\begin{proof}
We give a slightly different proof to that in \cite{ashcroftmermin} for the convenience of the reader. If $\Pi(\n)$ is a lattice plane then $d^{-1}\n\cdot\bv\in\Zed$ for all $\bv\in\mL(\B)$, so that $d^{-1}\n\in\mL(\B^{-T})$, and there is no shorter such dual lattice vector.

Conversely, if $\kv\in\mL(\B^{-T})\setminus\{\bzero\}$ with $\kv$ parallel to $\n$ then $\kv\cdot\B\e_i=n_i\in\Zed$ for each $i$, and hence $\B^T\kv=\sum_{i=1}^3n_i\e_i:=\n\in\Zed^3$. Pick $\m,\m'\in\Zed^3\setminus\{\bzero\}$ with $\m\cdot\n=\m'\cdot\n=0$ and $\m,\m'$ linearly independent, which is possible because the vectors $\n\wedge\e_i$ are not all parallel. Then $\B^T\kv\cdot\m=\B^T\kv\cdot\m'=0$, so that $\kv\cdot\B\m=\kv\cdot\B\m'=0$ and $\Pi(\n)$ is a lattice plane.\qed
\end{proof}
\begin{remark}
\label{rem1}
Suppose $\kv=\B^{-T}\q$ is parallel to $\n$, where $\q=(q_1,q_2,q_3)\in\Zed^3\setminus\{\bzero\}$. Then the formula for $d$ can be rewritten as 
\be
\label{interplane}d^{-1}=\frac{|\B^{-T}\q|}{\text{gcd}(q_1,q_2,q_3)},
\ee
where $\text{gcd}(q_1,q_2,q_3)$ is the greatest common divisor of the $q_i$. Indeed  $$\displaystyle\bar\q:=\frac{\q}{\text{gcd}(q_1,q_2,q_3)}\in \Zed^3$$ and the $\bar q_i$ have no common factor. If $\hat \q\in\Zed^3$ is parallel to $\q$ we have that  $n\hat \q=m\bar \q$ for  coprime integers $n\geq 1$ and $m$. Thus $n=1$ and so $|\bar\q|\leq |\hat\q|$ and thus $|\B^{-T}\hat\q|\geq |\B^{-T}\bar\q|$.
\end{remark}

\subsection{Cubic lattices}
As examples we consider \\ \noindent
(i) the {\it simple cubic} lattice, for which $\B=\B_{\rm c}$, where
\be
\label{0}
\B_{\rm c}=\Q\in {\rm O}(3).
\ee
 Equivalently 
\be
\label{00}\B_{\rm c}^T\B_{\rm c}=\1.
\ee\noindent
\begin{figure}[htbp] 
  \centering
  \includegraphics[width=4.67in,height=2.16in,keepaspectratio]{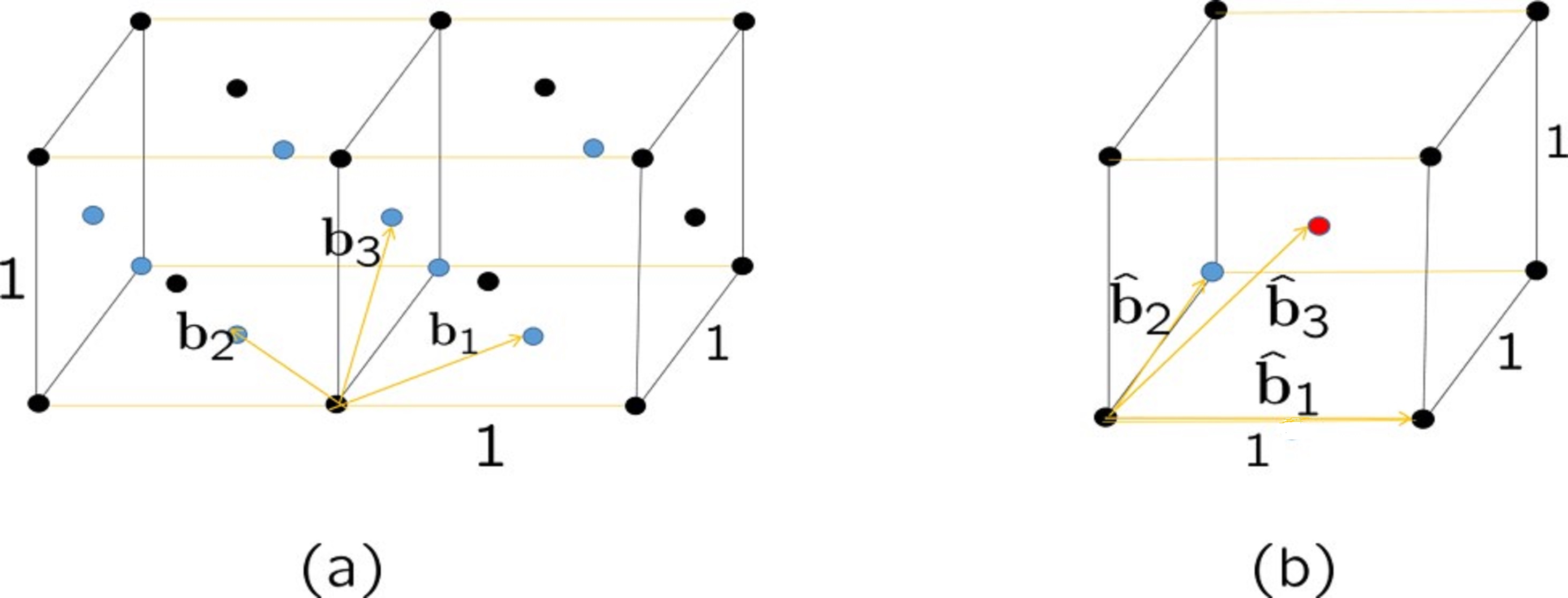}
  \caption{(a) face-centred cubic lattice, (b) body-centred cubic lattice}
  \label{fig:fccbcc}
\end{figure}\noindent
(ii)  the {\it face-centred cubic} (fcc) lattice (see Fig. \ref{fig:fccbcc} (a))  with the indicated basis vectors, for which $\B=\B_{\rm fcc}$, where
 \be \B_{\rm fcc}=\half\Q\left(\begin{array}{rrr}1&-1&\,\,0\\1&1&\,\,1\\0&0&\,\,1\end{array}\right), \Q\in {\rm O}(3).\ee
 Equivalently
\be
\label{1}\B_{\rm fcc}^T\B_{\rm fcc}=\frac{1}{4}\left(\begin{array}{ccc}2\,\,&0\,\,\,&1\\0\,\,\,&2\,\,\,&1\\1\,\,\,&1\,\,\,&2\end{array}\right),
\ee
(iii) the {\it body-centred cubic} (bcc) lattice (see Fig. \ref{fig:fccbcc} (b)), for which one could take the basis vectors  $\hat\bv_i$ shown, but the usual  choice is $\B=\B_{\rm bcc}$, where
\be
\label{2}\B_{\rm bcc}=\frac{1}{2}\Q\left(\begin{array}{rrr}-1&1&1\\ 1&-1&1\\ 1&1&-1\end{array}\right), \Q\in O(3),\ee
 for which
\be
\label{3}\B_{\rm bcc}^T\B_{\rm bcc}=\frac{1}{4}\left(\begin{array}{rrr}3&-1&-1\\ -1&3&-1\\-1&-1&3\end{array}\right).
\ee
Below we always take $\Q=\1$ for these lattices.
\begin{remark}
\label{bccfcc0}\rm
It follows from Theorem \ref{equiv} that for fcc or bcc lattices it is not possible to choose lattice vectors such that $\B\in\Zed^{3\times 3}$. Taking the case of fcc, for example, this would imply that
\be
\label{fccbcc1}
\half\left(\begin{array}{rrr}1&-1&\,\,0\\1&1&\,\,1\\0&0&\,\,1\end{array}\right)\muv\in\Zed^{3\times  3}
\ee
for some $\muv\in GL(3,\Zed)$, so that $\half \mu_{3j}\in \Zed$ for each $j$, and hence $\mu_{3j}$ is even, implying that $\det \muv$ is even, a contradiction.
\end{remark}

 \begin{remark}
\label{bccfccdual}
We have that $\B_{\rm fcc}^T\B_{\rm bcc}=\half{\bm \omega}^{-1}$, where
\be
\label{omegadef} \omv=\left(\begin{array}{rrr}0&0&\,\,1\\ 0&-1&\,\,1\\ 1&0&\,0\end{array}\right)\in GL^+(3,\mathbb Z).
\ee 
Hence $\B_{\rm fcc}^{-T}=2\B_{\rm bcc}\bm \omega$, proving the well-known result that   the dual lattice to fcc {is} 2bcc.
\end{remark}

The simple cubic, bcc and fcc lattices all have the same point group $P^{48}$ consisting of the 48 matrices $\Q\in\Othree$ mapping $\Zed^3$ to itself (equivalently the matrices with a single entry $\pm 1$ in each row and column). We write 
\be
\label{P24}
P^{24}=P^{48}\cap \SOthree
\ee
 for the 24 such matrices $\Q$ with $\det\Q=1$.

\section{The Ericksen energy well picture}
\label{ericksen}
Suppose that the free energy per unit volume of a crystalline material with atoms at the points of the Bravais lattice $\mathcal L(\C)$, where $\C\in GL^+(3,\R)$,  is given by $\varphi(\C)\geq 0$. For simplicity we assume that the temperature is constant, so that there are no phase transformations. 

Natural requirements on $\varphi$ are \\
(i) ({\it frame-indifference}) $$\varphi(\Q\C)=\varphi(\C)\text{ for all  }\Q\in \SOthree,$$
(ii) ({\it invariance with respect to equivalent lattices}) $$\varphi(\C{\bm\muv})=\varphi(\C) \text{ for all } {\bm\muv}\in GL^+(3,\mathbb Z).$$
We assume further that $\varphi(\C)=0$ iff $\C=\Q\B\muv$ for $\Q\in \SOthree$, $\muv\in GL^+(3,\mathbb Z)$, so that the minimum value zero is attained only for the unstrained lattice.

As is standard, we use the Cauchy-Born rule to relate the macroscopic free-energy density $\psi$ to $\varphi$, thus defining an elastic free energy
$$I(\y)=\int_\Omega \psi(D\y(\x))\,d\x$$
for a deformation $\y:\Omega\to\R^3$, where $\om\subset \R^3$ is an open reference domain. 

 Choosing a reference configuration in which the crystal lattice is $\mathcal L(\B)$, where  $\B\in GL^+(3,\R)$,  we assume that 
$$\psi(\A)=\varphi(\A\B),\mbox{ for }\A\in GL^+(3,\R).$$
Thus $\psi\geq 0$ inherits from $\varphi$  the invariances\\
(i) $\psi(\Q\A)=\psi(\A)$ for all $\Q\in \SOthree$,\\
(ii) $\psi(\A\B{\bm \mu}\B^{-1})=\psi(\A)$ for all ${\bm\mu}\in GL^+(3,{\mathbb Z})$,\\
and $\psi$ has zero set
\be
\label{zeroset}\psi^{-1}(0)=\bigcup_{\muv\in GL^+(3,\mathbb Z)}\SOthree\B\muv\B^{-1}.
\ee

The energy wells $\SOthree\B\muv\B^{-1}$ are not all distinct, because $\SOthree\B\muv\B^{-1}=\SOthree\B\tilde\muv\B^{-1}$ iff
$$\B\muv\tilde\muv^{-1}\B^{-1}\in P(\B)\cap \SOthree.$$
But since $P(\B)$ is finite and $GL^+(3,\mathbb Z)$ is infinite, there are {\it infinitely many distinct energy wells}.

\section{Interfaces}
\label{interfaces}
\subsection{Rank-one connections}
\label{rone}
We are interested in possible planar interfaces between distinct constant gradients on the energy wells, that is in pairs of distinct matrices $\F,\G\in\psi^{-1}(0)$ with 
\be
\label{4}\F-\G=\av\otimes\n,\;\;\av,\n\in\R^3,|\n|=1.
\ee
We can assume that $\G=\1$, so that we are interested in the $\muv\in GL^+(3,\mathbb Z)$ with $\M=\B\muv\B^{-1}\not\in \SOthree$ such that 
\be
\label{4a}
\1+\av\otimes\n=\Rv \M \text{ for some }\av, \n\in\R^3, |\n|=1, \Rv\in \SOthree.
\ee
Note that since $\M$ is independent of the scale of $\mL(\B)$ (that is it is the same for $\mL(t\B)$ for any $t\neq 0$) the rank-one connections are also independent of scale. If \eqref{4a} holds then since $\det\M=1$ we have that 
\be
\label{adotn}
\av\cdot\n=0, 
\ee
and thus
 \be
\label{4a1} |\av|^2=|\M|^2-3.
\ee
We denote by $0<\lambda_1\leq\lambda_2\leq\lambda_3$ the eigenvalues of the positive definite symmetric matrix $\M^T\M$, which satisfy $\lambda_1\lambda_2\lambda_3=1$, and by $\hat\e_i$ the corresponding orthonormal eigenvectors, so that $\M^T\M$ has spectral decomposition
\be
\label{4d}
\M^T\M=\lambda_1\hat\e_1\otimes\hat\e_1+\lambda_2\hat\e_2\otimes\hat\e_2+\lambda_3\hat\e_3\otimes\hat\e_3.
\ee
  A necessary and sufficient condition for \eqref{4a} to hold  (see e.g. \cite[Prop. 4]{j32}, \cite[Theorem 2.1]{j56}, \cite{Khachaturyan83}, \cite[Lemma 1]{p37}) is that $\lambda_2=1$, or  (since $\M^T\M\neq\1$) that $\M^T\M$ has an eigenvalue equal to one. Then $0<\lambda_1<1=\lambda_2<\lambda_3=\lambda_1^{-1}$. An equivalent condition   is that 
\be
\label{4b}
|\M|^2=|\M^{-1}|^2,
\ee
which follows from the identity
\be
\label{4c}
\det(\M^T\M-\1)= (\det \M)^2-(\det\M)^2|\M^{-1}|^2+|\M|^2-1.
\ee
(This implies in particular, as observed in \cite[Remark 8.4]{pitterizanzotto03},
 that there is a rank-one connection  to $\SOthree \M$ iff there is a rank-one connection to $\SOthree\M^{-1}$, which can be verified directly.)
 
  There are then exactly two distinct such {\it conjugate}  (or {\it reciprocal}) rank-one connections
\begin{eqnarray*}
\1+\av_+\otimes\n_+&=&\Rv_+\M,\\
\1+\av_-\otimes\n_-&=&\Rv_-\M.
\end{eqnarray*}
These can be calculated explicitly by comparing the spectral decomposition \eqref{4d} with the relation
\be
\label{4d1}
\M^T\M=(\1+\n\otimes\av)(\1+\av\otimes\n)
\ee
that follows from \eqref{4a} (see e.g.  \cite[Prop. 4]{j32}).
A straightforward calculation then shows that (up to a change of sign)
\be
\label{4e}
\n_\pm=\frac{1}{\sqrt{1+\lambda_1}}(\pm\sqrt{\lambda_1}\hat\e_1+\hat\e_3),\;\;\av_\pm=\frac{1-\lambda_1}{\sqrt{1+\lambda_1}}\left(\mp\frac{1}{\sqrt{\lambda_1}}\hat\e_1+\hat\e_3\right).
\ee
From \eqref{4a1} or \eqref{4e} we have that $|\av_+|=|\av_-|$. Also $\displaystyle\n_+\cdot\n_-=\frac{1-\lambda_1}{1+\lambda_1}$, so that the planes corresponding to the two rank-one connections are never orthogonal. 

It also follows from \eqref{4e} that if we know one rank-one connection $\1+\av\otimes \n=\Rv \M$ then the conjugate rank-one connection $\1+\tilde\av\otimes\tilde\n=\tilde\Rv\M$ is given by
\be
\label{4e1}
\tilde \av\otimes\tilde\n=\frac{1}{4+|\av|^2}(2\n-\av)\otimes(2\av+|\av|^2\n).
\ee
From \eqref{4e1} we have that
\be
\label{4e1a}
\tilde \av\otimes\tilde \av=(\tilde \av\otimes \tilde \n)(\tilde \n\otimes\tilde\av)=\frac{|\av|^2}{4+|\av|^2}(2\n-\av)\otimes (2\n-\av), 
\ee
from which, recalling from \eqref{adotn},\eqref{4a1} that $\av\cdot\n=0$ and $|\tilde\av|^2=|\av|^2$, we deduce that
\be
\label{4e1b}
2\n\otimes\n-\av\otimes\n=2\frac{\tilde\av}{|\tilde\av|}\otimes\frac{\tilde\av}{|\tilde\av|}+\tilde\av\otimes\tilde\n.
\ee
Hence we obtain the relations
\be
\label{4e2}
(-\1+2\n\otimes\n)(\1+\av\otimes\n)&=&\left(-\1+2\frac{\tilde\av}{|\tilde\av|}\otimes\frac{\tilde\av}{|\tilde\av|}\right)(\1+\tilde\av\otimes\tilde\n)\\
(-\1+2\tilde\n\otimes\tilde\n)(\1+\tilde\av\otimes\tilde\n)&=&\left(-\1+2\frac{\av}{|\av|}\otimes\frac{\av}{|\av|}\right)(\1+\av\otimes\n),\label{4e3}
\ee
where \eqref{4e3} is obtained from \eqref{4e2} by interchanging  $\av,\n$ and $\tilde\av,\tilde\n$. Thus $\tilde\Rv \Rv^T$ can be expressed as the product of two $180^\circ$ rotations in the two ways
\be
\label{4e4}
\tilde\Rv\Rv^T&=&(-\1+2\tilde\n\otimes\tilde\n)\left(-\1+2\frac{\av}{|\av|}\otimes\frac{\av}{|\av|}\right)\\
&=&\left(-\1+2\frac{\tilde\av}{|\tilde\av|}\otimes\frac{\tilde\av}{|\tilde\av|}\right)(-\1+2\n\otimes\n).
\nonumber
\ee
We will see that there are always $\muv$ such that $\M^T\M$ has an eigenvalue one. However this is not true for general $\muv$.

\subsection{Correspondence between rank-one connections for Bravais lattices and their duals}
It turns out that the rank-one connections for the dual of a Bravais lattice  can be obtained explicitly in terms of the rank-one connections for the original lattice (and vice versa). On account of Remark \ref{bccfccdual} this implies that  the rank-one connections for fcc can be obtained explicitly in terms of those for bcc, which are a bit easier to calculate due to the more symmetric form of the matrix $\B_{\rm bcc}$.

Note that the energy wells for the dual lattice $\mL(\B^{-T})$ of the Bravais lattice $\mL(\B)$ are given by $\SOthree \B^{-T}\muv^T\B^T$ for $\muv\in GL^+(3,\Zed)$.
\begin{theorem}
\label{bccfcc}The rank-one connection
\be
\label{bccfcc1}
\1 +\av\otimes\n=\Rv \B\muv\B^{-1},
\ee
for $\Rv\in {\rm SO}(3)$ and $\muv\in GL^+(3,\Zed)$ holds  iff
\be 
\label{bccfcc2}
\1+\Rv^T\n\otimes\Rv^T\av=\Rv^T\B^{-T}\muv^T\B^T.
\ee
Furthermore, if $\1+\tilde\av\otimes \tilde\n$ is the conjugate of $\1+\av\otimes\n$ then $\1+\Rv^T\tilde\n\otimes \Rv^T\tilde\av$ is the conjugate of $\1+\Rv^T\n\otimes\Rv^T\av$. 
\end{theorem}
\begin{proof}
Taking the transpose of \eqref{bccfcc1} and pre- and post-multiplying by $\Rv^T$ and $\Rv$ respectively, we see that \eqref{bccfcc1} and \eqref{bccfcc2} are equivalent. 

The statement about the conjugates follows from \eqref{4e1}, or more directly by noting that 
\be
\label{bccfcc6}
(\1+\tilde\av\otimes\tilde\n)(\1+\av\otimes\n)^{-1}=\Rv\in {\rm SO}(3),
\ee
where $\Rv\neq\1$, implies that 
\be
\label{bccfcc7}
(\1+\Rv^T\n\otimes\Rv^T\av)^{-1}(\1+\Rv^T\tilde\n\otimes\Rv^T\tilde\av)=\Rv^T\Rv^T\Rv\neq\1,
\ee
as required.
\qed
\end{proof}
From Remark \ref{bccfccdual} we thus obtain
\begin{corollary}
\label{fccbcccor}The rank-one connection
\be
\label{bccfcc7a}
\1+\av\otimes\n=\Rv\B_{\rm fcc}\muv\B_{\rm fcc}^{-1}
\ee
for $\Rv\in \SOthree$ and $\muv\in GL^+(3,\Zed)$ holds iff
\be
\label{bccfcc5}
\1+\Rv^T\n\otimes\Rv^T\av=\Rv^T\B_{\rm bcc}\tilde\muv\B_{\rm bcc}^{-1},
\ee
where $\tilde\muv=\omv\muv^T\omv^{-1}$ and $\omv$ is given by \eqref{omegadef}.
\end{corollary}

\section{Slip}
\label{slip}
\subsection{Slip systems and lattice invariant shears}
A {\it slip system} $(\Pi(\n),\bv)$ consists of a lattice plane $\Pi(\n)$ with unit normal $\n$ and interplane spacing $d$, and a nonzero lattice vector $\bv\in{\mathcal L}(\B)$ (the {\it Burgers vector}) with $\bv\cdot\n=0$. If $(\Pi(\n),\bv)$ is a slip system so is $(\Pi(\n),-\bv)$, which corresponds to an opposite direction of shear on the same lattice plane, and when counting slip systems we identify  $(\Pi(\n),\bv)$ with $(\Pi(\n),-\bv)$.

Equivalently
$${\mathcal L}(\B)=({\mathcal L}(\B)\cap \{\x\cdot\n\leq 0\})\cup(({\mathcal L}(\B)\cap\{\x\cdot\n>0\})+\bv),$$
so that if the part of ${\mathcal L}(\B)$ in the half-space $\{\x\cdot\n>0\}$ is rigidly displaced by $\bv$ then ${\mathcal L}(\B)$ is restored.

\begin{theorem} 
\label{slipchar}
The following are equivalent:\\
$\rm(i)$ $\1+\av\otimes\n=\B\muv\B^{-1},$ where $\av,\n\in\R^3$, $|\n|=1$,  and $\muv\in GL^+(3,\Zed)$.\\
$\rm(ii)$ $\muv=\1+\p\otimes\q$, where $\p,\q\in \Zed^3\setminus\{\bzero\}$, $\p\cdot\q=0$, and
\be
\label{z1}
\n=\tau\frac{\B^{-T}\q}{|\B^{-T}\q|},\;\;\av=\tau\B\p|\B^{-T}\q|,
\ee
where $\tau=\pm 1$.\\
$\rm(iii)$ 
 $(\Pi(\n),\bv)$ is a slip system with interplane spacing $d$  and Burgers vector $\bv=d\av$.
\end{theorem}
The proof uses the following lemma.
\begin{lemma}
\label{lemtensor}The  matrix
$\cv\otimes\dv\in \Zed^{3\times 3}$ for some $\cv,\dv\in\R^3$ if and only if $\cv\otimes\dv=\p\otimes\q$ for some $\p,\q\in\Zed^3$.
\end{lemma}
\begin{proof}
We only have to prove the necessity, and  may assume $\cv\otimes\dv$ is a nonzero matrix of integers. Thus some $c_i\neq 0$ and $c_id_1, c_id_2, c_id_3$ are integers, so that $\displaystyle\dv=\frac{1}{c_i}\m$ for some $\m\in\Zed^3$. Similarly $\displaystyle\cv=\frac{1}{d_j}\rv$ for some $j$ and $\rv\in\Zed^3$. Hence $\displaystyle\cv\otimes \dv=\frac{1}{c_id_j}\rv\otimes\m$. We can write $\rv=k\rv', \m=l\m'$, where $k,l$ are integers and neither $(r_1',r_2',r_3')$ nor $(m_1',m_2',m_3')$ have common factors. Thus  
$$\cv\otimes \dv=\frac{r}{s}\rv'\otimes \m',$$
where $\displaystyle\frac{r}{s}$ is a rational number expressed in its lowest terms. Suppose $s\neq\pm 1$, so that $s$ has a prime factor $S$. For some $r_i'$, $S$ does not divide $r_i'$. Hence $S$ divides $m_j'$ for all $j$, a contradiction. Hence $s=\pm 1$, giving the result.\qed
\end{proof}
\noindent {\it Proof of Theorem {\rm\ref{slipchar}}}\\
\noindent(i)$\Rightarrow$(ii). (i) implies that $\1+\B^{-1}\av\otimes\B^{T}\n=\muv$. Hence, by Lemma \ref{lemtensor}, $\av\otimes\n=\B\p\otimes\B^{-T}\q$ for $\p,\q\in\Zed^3$, and $\av\cdot\n=\p\cdot\q=0$, giving (ii).\vspace{.1in}

 \noindent(ii)$\Rightarrow$(iii).  We can suppose that the $q_i$ have no common factor, so that by Remark \ref{rem1},
$\displaystyle\n=\tau\frac{\B^{-T}\q}{|\B^{-T}\q|}$ is the normal to a lattice plane $\Pi(\n)$ with interplane spacing $\displaystyle d=\frac{1}{|\B^{-T}\q|}$. Also $\bv=\tau\B\p\in\mL(\B)$ with $\bv\cdot\n=0$, and so $\bv=d\av$, giving (iii). \vspace{.1in}

\noindent(iii)$\Rightarrow$(i). By Theorem \ref{lplane} and Remark \ref{rem1} we know that $\displaystyle\n=\frac{\B^{-T}\q}{|\B^{-T}\q|}, d=\frac{1}{|\B^{-T}\q|}$, for some $\q\in\Zed^3\setminus\{\bzero\}$. Also $\bv=\B\p$ for some $\p\in\Zed^3\setminus\{\bzero\}$, where $\bv\cdot\n=d\B\p\cdot\B^{-T}\q=d\p\cdot\q=0$. Then
$$\1+d^{-1}\bv\otimes\n=\1+\B\p\otimes\B^{-T}\q=\B(\1+\p\otimes\q)\B^{-1},$$
giving (i).\qed

\noindent Thus a slip system corresponds to a {\it lattice invariant shear} giving a rank-one connection between $\SOthree$ and $\SOthree\B\muv\B^{-1}$ without lattice rotation. The shear can correspond to a shear band, or in the case of rigid slip to a `microshear' over just one pair of adjacent lattice planes. Note that Theorem \ref{slipchar} shows that the slip systems for any Bravais lattice are in 1-1 correspondence with those for the simple cubic lattice (having deformation gradient $\1+\p\otimes\q$). For the use of such rank-one connections in continuum theories of plasticity see e.g. Ortiz \& Repetto \cite{ortizrepetto1999}.

\subsection{Slips of minimum shear magnitude}
\label{slipsmin}
Since if $\muv=\1+\p\otimes\q$ ,
$$\M=\B\muv\B^{-1}=\1+\B\p\otimes\B^{-T}\q$$
we have that 
$|\av|^2=d^{-2}|\bv|^2=|\B\p|^2|\B^{-T}\q|^2.$
Hence, to minimize the shear magnitude $|\av|$ we need to minimize $|\B\p|^2|\B^{-T}\q|^2$ subject to $\p,\q\in{\mathbb Z}^3\setminus \{0\}$ with $\p\cdot\q=0$.

 Minimizing the shear magnitude is the criterion proposed by Chalmers \& Martius \cite{chalmers1952slip}, who give a heuristic justification in terms of energetics,  pointing out that the criterion of minimizing the magnitude of the Burgers vector $\bv$  does not determine the preferred slip planes. A review of different possible criteria is given in \cite{cordier2002}. 

 In the following theorem we give the slips of minimum shear magnitude  for the simple cubic, bcc and fcc lattices, reproducing in our framework classical results for the slip systems for these lattices (see, for example, \cite{jackson2012handbook}). See below for further comments.

\begin{theorem}
\label{slipthm}
The slips $\1+\av\otimes \n$ minimizing $|\av|$ are given by:\vspace{.05in}\\
$\rm(a)$ Simple cubic: $\av\otimes\n=\pm\e_i\otimes\e_j$, where $i\neq j$, for which $|\av|^2= 1$.\vspace{.05in}\\
$\rm(b)$ bcc: \vspace{-.32in}
\begin{align*}\av\otimes\n&=\half(\e_1+\e_2+\e_3)\otimes(\e_j-\e_k),\\
\av\otimes\n&=\half(\e_j+\e_k-\e_i)\otimes (\e_j-\e_k),\\
\av\otimes\n&=\pm\half(\e_j+\e_k-\e_i)\otimes(\e_i+\e_k),
\end{align*}
where $i,j,k\in\{1,2,3\}$ are distinct, for which $|\av|^2=\frac{3}{2}$.\vspace{.05in}\\
$\rm(c)$ fcc: \vspace{-.32in}
\begin{align*}\av\otimes\n&=\half (\e_j-\e_k)\otimes(\e_1+\e_2+\e_3),\\  \av\otimes\n&=\half(\e_j-\e_k)\otimes(\e_j+\e_k-\e_i),  \\\av\otimes\n&=\pm\half(\e_i+\e_k)\otimes(\e_j+\e_k-\e_i),
\end{align*}
 where $i,j,k\in\{1,2,3\}$ are distinct, for which $|\av|^2=\frac{3}{2}$.
\end{theorem}
\begin{remark}
\label{Prem}
Without loss of generality we may take for simple cubic $\av\otimes\n=\e_1\otimes\e_2$,
for bcc $\av\otimes\n=\half(\e_1+\e_2+\e_3)\otimes(\e_1-\e_2)$, and for fcc $\av\otimes\n=\half(\e_1-\e_2)\otimes(\e_1+\e_2+\e_3)$, since all the other cases are given by $\Pv\av\otimes \n\Pv^{-1}$ for suitable $\Pv\in P^{24}$, where $P^{24}$ is defined in \eqref{P24}.
\end{remark}
\begin{proof}
(a) Since $\B_{\rm c}=\1$ we have $|\av|^2=|\p|^2|\q|^2\geq 1$, and $|\av|^2=1$ for $|\p|^2=|\q|^2=1$, giving the result.\\

\noindent(b) For bcc 
$$\B_{\rm bcc}=\half\left(\begin{array}{rrr}-1&1&1\\ 1&-1&1\\ 1&1&-1\end{array}\right),\;\; \B_{\rm bcc}^{-1}=\left(\begin{array}{ccc}0\,\,&1\,\,&1\\ 1\,\,&0\,\,&1\\ 1\,\,&1\,\,&0\end{array}\right),$$
 and we have to minimize $\frac{1}{4}G(\p)H(\q)$,  subject to $\p,\q\in{\mathbb Z}^3\setminus \{0\}$ and $\p\cdot\q=0$, where
\begin{eqnarray*}&&G(\p)=(p_2+p_3-p_1)^2+(p_3+p_1-p_2)^2+(p_1+p_2-p_3)^2,\\ &&H(\q)=(q_2+q_3)^2+(q_3+q_1)^2+(q_1+q_2)^2.
\end{eqnarray*}
The minimum value of $H(\q)$ is 2, with minimizers $\q=\pm\e_i$ and $\q=\e_i-\e_j, i\neq j$, while the minimum value of $G(\p)$ is 3, with minimizers $\p=\pm\e_i$ and $\p=\pm(\e_1+\e_2+\e_3)$. Hence the minimizers of $\frac{1}{4}G(\p)H(\q)$,  subject to $\p,\q\in{\mathbb Z}^3\setminus \{0\}$ and $\p\cdot\q=0$ are given by 
\be\nonumber
\pm\p\otimes\q\in\{\e_i\otimes\e_j, \e_i\otimes(\e_j-\e_k), (\e_1+\e_2+\e_3)\otimes(\e_i-\e_j), i,j,k \text{ distinct}\}.
\ee
Calculating $\displaystyle\B_{\rm bcc}\p\otimes\B_{\rm bcc}^{-T}\q$ for these  possibilities gives the slip systems in the theorem.\\
(c) For fcc   it suffices to note that by Corollary \ref{fccbcccor} the slips  minimizing $|\av|$ are given by $\1+\n\otimes\av$, where $\1+\av\otimes\n$ are the slips for bcc minimizing $|\av|$.\qed
\end{proof}
 
For bcc metals the experimental determination of slip planes is not straightforward. One can ask what the slips are that give the next lowest value of $|\av|^2$, and whether such slips are observed. Since $H(\q)$ is even, the next lowest value of $G(\p)H(\q)$ is greater than or equal to 8, for which $|\av|^2=2$, and the value 8 is achieved with $G(\p)=4, H(\q)=2$ and
$$\p\otimes\q\in \{\pm(\e_i+\e_j)\otimes\e_k, (\e_i+\e_j)\otimes(\e_i-\e_j), i,j,k \text{ distinct}\}, $$
giving the slips
\be
\label{gh2}
\av\otimes\n\in\{\pm \e_k\otimes(\e_i+\e_j), \e_k\otimes(\e_i-\e_j), i,j,k\text{ distinct}\}.
\ee
These slips are consistent with those  described by Chalmers \& Martius \cite{chalmers1952slip}, who noted that they had not been reported experimentally; in a more recent survey   Weinberger, Boyce \& Bataille \cite{weinberger2013slip} appear to confirm this, stating that ``{\it It is clear from numerous experiments 
that slip occurs in the closest packed $\langle 111\rangle$ direction
and the Burgers vector is $\half\langle 111\rangle$}''. The slips in \eqref{gh2} have Burgers vector magnitude $|\bv|^2=1$, whereas those in Theorem \ref{slipthm}(b) have $|\bv|^2=\frac{3}{4}$.

It is often stated (see, for example, \cite{jackson2012handbook}) that there are up to 48 slip systems for bcc, consisting of the 12 systems given in Theorem \ref{slipthm} (identifying $\pm\av\otimes\n$) together with 12 having $(112)$ slip planes 
given by 
\be
\label{gh3}
\av\otimes\n\in\left\{\pm\half(\kappa_1\e_i+\kappa_2\e_j-\e_k)\otimes(\kappa_1\e_i-\kappa_2\e_j+2\e_k),\right.&&\\
&&\hspace{-1in}\left. i,j,k \text{ distinct}, \kappa_1=\pm 1, \kappa_2=\pm 1\right\},\nonumber
\ee
with $|\av|^2=\frac{9}{2}$,
  and 24 having $(123)$ slip planes, given by
\be
\label{gh4}
\av\otimes\n\in\left\{\pm\half(\e_i-\kappa_1\e_j-\kappa_2\e_k)\otimes(3\e_i+2\kappa_1\e_j +\kappa_2\e_k),\right.&&\\
&&\hspace{-1in}\left.   i,j,k \text{ distinct}, \kappa_1=\pm 1, \kappa_2=\pm 1\right\},
\nonumber
\ee
with $|\av|^2=\frac{21}{2}$. That these are slips can be checked by verifying that $\av\cdot\n=0$ and $\B^{-1}\av\otimes\B^T\n\in\Zed^3\otimes\Zed^3$. 
These extra slip systems have the same Burgers vector magnitude $|\bv|^2=\frac{3}{4}$ as those in Theorem \ref{slipthm}(b), but considerably higher values of shear magnitude.

For further discussion of the Chalmers \& Martius criterion for bcc crystals see 
\cite{chenmaddin}.

\section{Twinning}
\label{twins}
\subsection{Twins}
\label{twindef}
 In this paper we define $\F=\1+\av\otimes\n$ to be a (mechanical)  {\it twin} if  the lattices $\mL(\B)$ and $\F\mL(\B)$ on either side of the interface are  non-trivially reflected with respect to each other, so that $\F=\1+\av\otimes\n$  is not a slip and satisfies for some unit vector $\m$ 
\be
\label{twin}
\F{\mathcal L}(\B)=(\1-2\m\otimes\m)\mL(\B)=(-\1+2\m\otimes\m)\mathcal L(\B),
\ee
where for the second equality we have used $\mL(\B)=\mL(-\B)$. (Note that if $\F$ is a slip then \eqref{twin} holds for any $180^\circ$ rotation $-\1+2\m\otimes\m \in P(\B)$.)
By Corollary \ref{cor2}, \eqref{twin} holds iff 
\be
\label{twina}
\F=(-\1+2\m\otimes\m)\B\muv\B^{-1}
\ee
 for some $\muv\in GL^+(3,\Zed)$. 

The definition given above of twins corresponds to that given by Christian \cite[p51] {christian2002}, except that we require that $\F$ is not a slip, and grew out of a discussion with R.D. James \cite{jamesprivate}. It covers the cases of Type 1 and Type 2 twins in the next section, but other possibilities too; on the other hand it is less general than that used e.g. by Pitteri \& Zanzotto \cite{pitterizanzotto03}, who give (pp 28,29) references to various different definitions. According to their terminology `nonconventional twins' correspond to rank-one connections  which are neither Type 1 nor Type 2 twins. In Section \ref{conjslip} we give different examples of such `nonconventional twins' which are, and which  are not, twins as defined in this paper. For a related discussion see \cite{chenetal13}.
  
\subsection{Type 1, Type 2 and compound twins}
\label{twintypes}
We shall be interested in {\it Type} 1 twins, for which $\m=\n$, so that the lattices are non-trivially  reflected with respect to the twin plane,   and thus
\be
\label{type1}
\F= (-\1+2\n\otimes\n)\B\muv\B^{-1}
\ee
for some   $\muv\in GL^+(3,\Zed)$, and in {\it Type} 2 twins, for which $\F$ satisfies \eqref{twina} with $\displaystyle\m=\frac{\av}{|\av|}$. From \eqref{4e2} it follows that if $\F=\1+\av\otimes\n$ is a Type 1 twin, then its conjugate $\1+\tilde\av\otimes\tilde\n$ satisfies
\be
\label{type1a}
\1+\tilde\av\otimes\tilde\n=\left(-\1+2\frac{\tilde\av}{|\tilde\av|}\otimes\frac{\tilde\av}{|\tilde\av|}\right)\B\muv\B^{-1},
\ee
and so (see \cite[p254]{pitterizanzotto03}) is a    Type 2 twin (unless it is a slip, which can happen -- see Theorem \ref{twinthm}(a)). Conversely the conjugate to a Type 2 twin, unless it is a slip, is a Type 1 twin. $\F$ is a {\it compound twin} if it is both a Type 1 and Type 2 twin, so that
\be
\label{compound}
\F= (-\1+2\n\otimes\n)\B\muv\B^{-1}=\left(-\1+2\frac{\av}{|\av|}\otimes\frac{\av}{|\av|}\right)\B\nuv\B^{-1}
\ee
for some $\nuv\in GL^+(3,\Zed)$. This holds iff
\be
\label{compound1}
(-\1+2\n\otimes\n)\left(-\1+2\frac{\av}{|\av|}\otimes\frac{\av}{|\av|}\right)\in P(\B),
\ee
or equivalently iff
\be
\label{compound2}
\B^{-1}\left(\1-2\n\otimes\n-2\frac{\av}{|\av|}\otimes\frac{\av}{|\av|}\right)\B\in\Zed^{3\times 3}.
\ee
We will use \eqref{compound2}  to check for compound twins; other criteria are given in \cite[Proposition 8.7]{pitterizanzotto03}, \cite{forclaz99}. Obviously the conjugate of a compound twin is either a slip or a compound twin.

From \eqref{type1}, to find the Type 1 twins (and hence their conjugate Type 2 twins) we need to find $\muv,\av,\n$ such that $\B\muv\B^{-1}\not\in SO(3)$ and
\be 
1+\av\otimes\n=(-\1+2\n\otimes\n)\B\muv\B^{-1}.
\ee
Equivalently
$\muv=-\1+\B^{-1}(2\n-\av)\otimes\B^T\n,$ and
hence, by Lemma \ref{lemtensor}, 
\be \B^{-1}(2\n-\av)\otimes\B^T\n=\p\otimes\q
\ee
 for some $\p,\q\in\mathbb Z^3$, and $\p\cdot\q=2$.
Thus we can suppose that 
\be 
\label{compound3}\n=\frac{\B^{-T}\q}{|\B^{-T}\q|},\;2\n-\av=|\B^{-T}\q|\B\p,
\ee
 so that 
\be 
\label{compound4}
|\av|^2= |\B\p|^2|\B^{-T}\q|^2-4.
\ee
It follows from \eqref{4e1},\eqref{compound3},\eqref{compound4} that the conjugate Type 2 twin is given by
\be
\label{compound8}
\1+\tilde\av\otimes\tilde\n=\1+\B\p\otimes\left(\B^{-T}\q-2\frac{\B\p}{|\B\p|^2}\right).
\ee

\subsection{Twins of minimum shear magnitude}
\label{twinmin}
From \eqref{compound4},  in order to minimize $|\av|^2$ we have to minimize $|\B\p|^2|\B^{-T}\q|^2$ among $\p,\q\in{\mathbb Z}^3$ with $\p\cdot\q=2$ and $|\B\p|^2|\B^{-T}\q|^2>4$. In the following theorem we give the Type 1/Type 2 twins minimizing the shear magnitude for the simple cubic, bcc and fcc lattices. Note that by Theorem \ref{bccfcc} and Corollary \ref{fccbcccor} the Type 1 (respectively Type 2) twins for fcc are given by  $\1-\n\otimes\av$, where  $\1+\av\otimes\n$ are the Type 2  (respectively Type 1) twins for bcc. The twins identified, and their shear magnitude, correspond to those given in Christian \& Mahajan \cite[Table 1]{christianmahajan1995}, which summarizes earlier  work of Jaswon \& Dove \cite{jaswondove1956,jaswondove1957,jaswondove1960}, Bilby \& Crocker\cite{bilbycrocker1965}, Bevis \& Crocker \cite{beviscrocker1968,beviscrocker1969}.
\begin{theorem}
\label{twinthm}
The Type $1$ and Type $2$ twins $\1+\av\otimes\n$ minimizing $|\av|$ are given by:\vspace{.05in}\\
$\rm(a)$ Simple cubic:  the compound twins \begin{align*}\av\otimes\n&=\frac{1}{5}(2\e_j-\kappa\e_i)\otimes(2\kappa\e_i+\e_j),
\end{align*}
where $i\neq j,\,\kappa=\pm 1$, whose conjugates are the slips $\tilde\av\otimes\tilde\n= \kappa \e_i\otimes\e_j$ in Theorem $\ref{slipthm}(a)$, for which $|\av|^2=1$.\vspace{.05in}\\ 
$\rm(b)$ bcc: the conjugate pairs
\begin{align*}\av\otimes\n&=\frac{1}{6}(\e_j-\e_k+\e_i)\otimes(2\e_j+\e_k-\e_i),\\\tilde\av\otimes\tilde\n&=\frac{1}{6}(\e_j+\e_k-\e_i)\otimes(2\e_j+\e_i-\e_k), 
\end{align*}
and
\begin{align*}
\av\otimes\n&=\frac{1}{6}(\e_1+\e_2+\e_3)\otimes (2\e_i-\e_j-\e_k),\\
\tilde\av\otimes\tilde\n&=\frac{1}{6}(\e_i-\e_j-\e_k)\otimes(2\e_i+\e_j+\e_k),
\end{align*}
where $i,j,k$ are distinct,   which are all  compound twins, and for which $|\av|^2=\half$.\vspace{.05in}\\
$\rm(c)$ fcc: the conjugate pairs\\
\begin{align*}\av\otimes\n&=\frac{1}{6}(-2\e_j-\e_k+\e_i)\otimes(\e_j-\e_k+\e_i),\\\tilde\av\otimes\tilde\n&= \frac{1}{6}(-2\e_j-\e_i+\e_k)\otimes(\e_j+\e_k-\e_i),
\end{align*}
and
\begin{align*}\av\otimes\n&=\frac{1}{6}(-2\e_i+\e_j+\e_k)\otimes(\e_1+\e_2+\e_3), \\\tilde\av\otimes\tilde\n&=\frac{1}{6}(2\e_i+\e_j+\e_k)\otimes(\e_j+\e_k-\e_i),
\end{align*}
where $i,j,k$ are distinct,     which are all compound twins, and for which $|\av|^2=\half$.
\end{theorem}
\begin{proof}
(a) Since $\B=\1$ we have to minimize $|\p|^2|\q|^2$ among $\p,\q\in \Zed^3$ subject to $\p\cdot\q=2$ and $|\p|^2|\q|^2>4$, since $|\p|^2|\q|^2=4$ implies $\av=0$. The minimum is given by $|\p|^2|\q|^2=5$ with minimizers
$$\p\otimes\q\in\{(2\e_j-\kappa\e_i)\otimes\e_j, \;\e_i\otimes (2\e_i+\kappa\e_j),\; i\neq j, \kappa=\pm 1\}.$$
From \eqref{compound3} we have that 
$$\av\otimes\n=\frac{1}{|\q|^2}(2\q-|\q|^2\p)\otimes\q,$$
which together with \eqref{4e1} gives the conjugate pairs in the statement. Note that $\av\otimes\n=\frac{1}{5}(2\e_j-\kappa\e_i)\otimes(2\kappa\e_i+\e_j)$ is not a slip because it does not belong to $\Zed^{3\times 3}$, and is a compound twin by the criterion \eqref{compound2}, since $2(\e_j\otimes\e_j+\e_i\otimes\e_i)\in\Zed^{3\times 3}$.\\

\noindent (b)  For bcc (see the proof of Theorem \ref{slipthm} (b)) we have to minimize $G(\p)H(\q)$ subject to $\p\cdot\q=2$ and $G(\p)H(\q)>16$, where
\begin{eqnarray*}&&G(\p)=(p_2+p_3-p_1)^2+(p_3+p_1-p_2)^2+(p_1+p_2-p_3)^2,\\ &&H(\q)=(q_2+q_3)^2+(q_3+q_1)^2+(q_1+q_2)^2.
\end{eqnarray*}
Since $H(\q)$ is even the minimum value is $\geq 18$ and we show that the value 18 is attained. Also $G(\p)\geq 3$, so that the possibilities are $G(\p)=9, H(\q)=2$, $G(\p)=3, H(\q)=6$. But $G(\p)=9$ implies either (i) that $p_i+p_j-p_k=\pm 3$, $p_i+p_k-p_j=p_j+p_k-p_i=0$ for $i,j,k$ distinct, so that $p_k=0$ and $2p_i=\pm 3$, which is impossible, or (ii) that $p_i+p_j-p_k=2\kappa$, $p_k+p_i-p_j=2\kappa'$, $p_j+p_k-p_i=\kappa''$ for $i,j,k$ distinct and $\kappa,\kappa',\kappa''=\pm 1$, giving $2p_k=2\kappa'+\kappa''$, again impossible. Hence we need only consider the case $G(\p)=3$, $H(\q)=6$.

It is easy to check that $G(\p)=3=1^2+1^2+1^2$ iff 
\be
\label{sbcc1}
\pm\p\in\{\e_i, \e_1+\e_2+\e_3, i=1,2,3\},
\ee
 and that $H(\q)=6=2^2+1^2+1^2$ iff 
\be
\label{sbcc2}
\pm\q\in\{\e_i+\e_j, 2\e_i-\e_j, 2\e_i-\e_j-\e_k, i,j,k \text{ distinct}\}.
\ee
From \eqref{sbcc1},\eqref{sbcc2} we find that $\p\cdot\q=2$ iff
\be
\label{sbcc3}
\p\otimes\q\in \{\e_i\otimes(2\e_i-\e_j), \e_i\otimes(2\e_i-\e_j-\e_k),&&\\
&&\mbox{ }\hspace{-1in} (\e_1+\e_2+\e_3)\otimes(\e_i+\e_j), i,j,k\text{ distinct}\}.\nonumber
\ee
Calculating $\av\otimes\n$ from \eqref{compound3}
for these $\p\otimes\q$ and using \eqref{4e1} we obtain the conjugate pairs in the statement, none of which belong to $\Zed^{3\times 3}$ and so are not slips.

To show that the twins are compound we note that from \eqref{compound3} and the relations $|\av|=\half$, $|\B_{\rm bcc}^{-T}\q|^2=6$, $|\B_{\rm bcc}\p|^2=\frac{3}{4}$ that
\be
\label{sbcc4}
2\B_{\rm bcc}^{-1}\left(\n\otimes\n+\frac{\av\otimes\av}{|\av|^2}\right)\B_{\rm bcc}&&\\&&\mbox{ }\hspace{-1in}=3\G^{-1}\q\otimes\q+24\p\otimes\G\p-8\p\otimes\q-8\G^{-1}\q\otimes\G\p,\nonumber
\ee
where $\G=\B_{\rm bcc}^T\B_{\rm bcc}.$ Since $\G^{-1}, 4\G\in\Zed^{3\times 3}$, the twins are compound by \eqref{compound2}.
\\

\noindent(c) The twins for fcc are obtained from those for bcc using Theorem \ref{bccfcc} and  Corollary \ref{fccbcccor} as described above, which also imply that the twins are compound.\qed
\end{proof}
\begin{remark}
For the cubic lattices we consider it is not in general true that Type 1 and Type 2 twins are compound. For example, we can take $\p=2\e_1-\e_3, \q=2\e_1+\e_2+2\e_3$, so that 
 for the simple cubic lattice we obtain from \eqref{compound3}
$$\av\otimes\n=\frac{1}{9}(-14\e_1+2\e_2+13\e_3)\otimes(2\e_1+\e_2+2\e_3),
$$ 
with conjugate 
$$\tilde\av\otimes\tilde\n=\frac{1}{5}(2\e_1-\e_3)\otimes(6\e_1+5\e_2+12\e_3).$$
Then none of $$\av\otimes\n,\;\ \tilde\av\otimes\tilde\n, \;2\left(\n\otimes\n+\frac{\av\otimes\av}{|\av|^2}\right)$$ belong to $\Zed^{3\times 3}$ so that, by \eqref{compound2}, $1+\av\otimes\n, \1+\tilde\av\otimes\tilde\n$ are a pair of Type 1/Type 2 twins that are not compound. For the same $\p,\q$ it can be checked that the   corresponding pair of conjugate bcc Type 1/Type 2 twins given by
\be
\av\otimes\n&=&\frac{1}{34}(57\e_1-9\e_2-45\e_3)\otimes(3\e_1+4\e_2+3\e_3),\\\tilde\av\otimes\tilde\n&=&\frac{3}{38}(-3\e_1+\e_2+3\e_3)\otimes(23\e_1+24\e_2+15\e_3),\nonumber
\ee
are also not compound.
\end{remark}
\section{General rank-one connections minimizing shear magnitude}
\label{gen}
We  ask what the rank-one connections $\F=\1+\av\otimes\n$ are between $\1$ and $\psi^{-1}(0)=\bigcup_{\muv\in GL^+(3,\mathbb Z)}\SOthree\B\muv\B^{-1}$ that minimize $|\av|^2$. For the simple cubic lattice we have a complete answer.
 \begin{theorem}
\label{twinmin1}
For the simple cubic lattice the rank-one connections $\1+\av\otimes\n=\Rv\muv$, $\Rv\in 
\SOthree$, $\muv\in GL^+(3,\mathbb Z)\setminus\SOthree$ that minimize $|\av|$ are given by the slips in Theorem $\ref{slipthm}(a)$ and the twins in Theorem $\ref{twinthm}(a)$, with the minimum value $|\av|^2=1$.
 \end{theorem}
\begin{proof}
If $\muv\in GL^+(3,\Zed)$ then $|\muv|^2\geq 3$ with $|\muv|^2=3$ iff   $\muv\in\SOthree$ (this follows either by noting that each row and column must contain a nonzero entry, and that if there is a single entry $=\pm 1$ then $\muv$ is a rotation, or by use of the AM$\geq$GM inequality applied to the squares of the singular values of $\muv$).
Since $|\muv|^2=3+|\av|^2$ is an integer, the minimum value of $|\av|^2$ for a rank-one connection is $\geq 1$, and so by Theorems \ref{slipthm}, \ref{twinthm} the minimum value is $|\av|^2=1$, with $|\muv|^2=4$.

But $|\muv|^2=4$ implies that each row and column of $\muv$ contains an entry $\pm1$ and that there is a single additional entry $\pm 1$. Thus
$$\muv=\Q+\kappa\e_k\otimes\e_j$$
for some $\Q\in P^{48}$ with $Q_{kj}=0$ and $\kappa=\pm 1$. Hence
$$\muv=\Q(\1+\kappa'\e_i\otimes\e_j),$$
where $\e_i=\kappa\kappa'\Q^T\e_k$, $\kappa'=\pm 1$, and thus $\e_i\cdot\e_j=0$, from which it follows that $\det\Q=1$ and hence $\Q\in\SOthree$. The corresponding rank-one connections are then given by $\1+\av\otimes\n=\Rv\Q(\1+\kappa'\e_i\otimes\e_j)$ and are thus those in Theorems \ref{slipthm}, \ref{twinthm}.\qed
\end{proof}
For bcc and fcc we have the following result.
\begin{theorem}
For bcc and fcc the squared magnitude $|\av|^2$ for any rank-one connection $\1+\av\otimes\n=\Rv\B\muv\B^{-1}\not\in\SOthree$, $\muv\in GL^+(3,\mathbb Z)$ is given by $|\av|^2=\frac{m}{2}$, where $m$ is a positive integer.
\end{theorem}
\begin{proof}
By Theorem \ref{bccfcc} it suffices to prove the result for bcc, for which we can write
\be
\label{gen1}
\B_{\rm bcc}=-\1+\half\vv\otimes\vv,\; \vv=\e_1+\e_2+\e_3.
\ee
Then, setting $\M=\B_{\rm bcc}\muv\B_{\rm bcc}^{-1}$, $\G=\B_{\rm bcc}^T\B_{\rm bcc}=\1-\frac{1}{4}\vv\otimes\vv$, and noting that $\G^{-1}=\1+\vv\otimes\vv$, we calculate
\be
\nonumber
|\M|^2&=&\tr(\G^{-1}\muv^T\G\muv)\\
&=&|\muv|^2+|\muv\vv|^2-\frac{1}{4}(|\muv^T\vv|^2+(\muv\vv\cdot\vv)^2).
\label{gen2}
\ee
We now note that 
\be
\nonumber
|\muv^T\vv|^2+(\muv\vv\cdot\vv)^2&=&c_1^2+c_2^2+c_3^2+(c_1+c_2+c_3)^2\\
&=&2(c_1^2+c_2^2+c_3^2+c_2c_3+c_3c_1+c_1c_2),\label{gen3}
\ee
where $c_i$ denotes the sum of the entries in the $i^{\rm th}$ column of $\muv$. Since $\muv\in\Zed^{3\times 3}$, $\vv\in\Zed^3$, and $|\av|^2=|\M|^2-3$, the result follows from \eqref{gen2}, \eqref{gen3}.\qed
\end{proof}
\begin{corollary}
\label{cor1}
For bcc and fcc the twins in Theorem $\ref{twinthm} (b), (c)$ minimize $|\av|^2$.
\end{corollary}
\begin{proof}
The twins satisfy $|\av|^2=\half$, the least possible value.\qed
\end{proof}
\begin{remark}
Corollary \ref{cor1} leaves unresolved the possibility that some other rank-one connection for bcc or fcc also gives the value $|\av|^2=\half$. This could be decided computationally, since there are only finitely many possibilities for $\muv$, for each of which the existence of rank-one connections could be checked. Indeed, if $|\M|^2=3+\half=\frac{7}{2}$ then $|\muv|^2\leq|\B_{\rm bcc}^{-1}|^2|\M|^2|\B_{\rm bcc}|^2=47.25$ so that $|\muv|^2\leq 47$. However we can get a better estimate by noting that     \eqref{gen2} can be rewritten as
\be\nonumber
|\M|^2
=\frac{1}{4}|\muv|^2+\frac{1}{4}|\muv\vv|^2+\frac{1}{12}|3\muv^T-\muv^T\vv\otimes\vv|^2+\frac{1}{12}|3\muv\vv-(\muv\vv\cdot\vv)\vv|^2\nonumber.
\ee
Since $\muv$ is nonsingular we have that  $|\muv\vv|^2\geq 1$. Also $\muv^T$ is not a rank-one matrix, so that $|3\muv^T-\muv^T\vv\otimes\vv|^2\geq 1$.  Hence $|\muv|^2\leq 13-\frac{1}{3}$ and so $|\muv|^2\leq 12$.
\end{remark}

\section{Conjugates of slips}
\label{conjslip}
We first ask when the conjugate of a slip is a slip.
\begin{theorem}
\label{slipslip}
For simple cubic, bcc and fcc the only slips whose conjugates are slips are given by the conjugate pairs 
\begin{align*}\av\otimes\n&= 2\kappa\e_i\otimes \e_j,\\ \tilde\av\otimes\tilde\n&=(\e_j-\kappa\e_i)\otimes (\kappa\e_i+\e_j)
\end{align*}
 with $i\neq j$ and $\kappa=\pm 1$, for which $|\av|^2=4$.
\end{theorem}
\begin{proof}
We need to find the slips $\1+\av\otimes\n$ such that 
\be
\label{ss20}
\Q&=&(\1+\tilde\av\otimes\tilde\n)(\1+\av\otimes\n)^{-1}\\
&=&\1-\frac{2|\av|^2}{4+|\av|^2}\left(\n\otimes\n+\frac{\av}{|\av|}\otimes\frac{\av}{|\av|}\right)+\frac{4}{4+|\av|^2}\left(\n\otimes\av-\av\otimes\n\right)\in P^{24},\nonumber
\ee
where we have used \eqref{4e1} and $(\1+\av\otimes\n)^{-1}=\1-\av\otimes\n$.

For this to be the case we need that 
\be
\label{ss21}\tr\Q=\frac{12-|\av|^2}{4+|\av|^2}=m\in\Zed,
\ee
so that $|\av|^2=\displaystyle 4\left(\frac{3-m}{1+m}\right)$, giving the possibilities $m=0,1$ or 2.

The case $m=2$ is impossible since no $\Q\in P^{24}$ has $\tr\Q=2$. The case $m=0$ gives $|\av|^2=12$ and
\be
\label{ss22}
\Q=\1-\frac{3}{2}(\n\otimes\n+\nuv\otimes\nuv)+\sqrt\frac{3}{2}(\n\otimes\nuv-\nuv\otimes\n),
\ee
where $\nuv=\displaystyle\frac{\av}{|\av|}$. But this is a rotation about $\nuv\wedge\n$ through an angle $\displaystyle\pm \frac{\pi}{3}$, and there are no such elements of $P^{24}$.

Thus we only have to consider the case $m=1$, for which $|\av|^2=4$ and
\be
\label{ss23}
\Q=\1-(\n\otimes\n+\nuv\otimes\nuv)+\n\otimes\nuv-\nuv\otimes\n.
\ee
This is a rotation through an angle $\pm\pi/2$ about $\nuv\wedge\n$, and the only such rotations in $P^{24}$ are those with axes $\e_i$, so that without loss of generality we can suppose that 
\be
\label{ss24}
\Q=\left(\begin{array}{ccc}0&-\kappa&0\\ \kappa&0&0\\0&0&1\end{array}\right),
\ee
where $\kappa=\pm 1$, and $\n=\cos\theta\,\e_1+\sin\theta\,\e_2$, $\nuv=\kappa(\sin\theta\,\e_1-\cos\theta\,\e_2)$.
For the simple cubic case we then require that 
\be
\label{ss25}
\av\otimes\n=\kappa \left(\begin{array}{ccc}\sin\,2\theta &1-\cos 2\theta&0\\
-1-\cos\,2\theta&-\sin\,2\theta&0\\0&0&\kappa\end{array}\right)\in \Zed^{3\times 3},
\ee
which holds iff $2\theta=k\frac{\pi}{2}$, giving the possibilities in the statement. Similarly, for bcc we require that $\B^{-1}\av\otimes\B^T\n\in\Zed^{3\times 3}$, which again holds iff $2\theta=k\frac{\pi}{2}$.\qed
\end{proof}
It is interesting to calculate the conjugates to the slips in Theorem \ref{slipthm} and  \eqref{gh2}- \eqref{gh4}. For the simple cubic case we have already seen (Theorem \ref{twinthm}(a)) that the rank-one connections conjugate to the slips in Theorem \ref{slipthm}(a) are compound twins. For bcc and fcc we can without loss of generality (using Theorem \ref{bccfcc}) take the case of bcc, when we have the following result.
\begin{theorem}
\label{bccslipconj}
The conjugates of the slips in Theorem $\ref{slipthm}(b)$ are twins which are neither Type $1$ nor Type $2$.
\end{theorem}
\begin{proof}
By Remark \ref{Prem} we need only consider the case
\be
\label{ss1}
\1+\av\otimes\n=\1+\half(\e_1+\e_2+\e_3)\otimes (\e_1-\e_2),
\ee
which by \eqref{4e1} has conjugate
\be
\label{ss2}
\1+\tilde\av\otimes\tilde\n=\1+\frac{1}{22}(\e_1-3\e_2-\e_3)\otimes(7\e_1+\e_2+4\e_3),
\ee
and is not  a slip by Theorem \ref{slipslip}, since $|\av|^2\neq 4$. 
Then the rotation (see \eqref{4e4}) $\Q=\tilde\Rv\Rv^T=(\1+\tilde\av\otimes\tilde\n)(\1+\av\otimes\n)^{-1}$ is given by
\be
\label{ss3}
\Q=\frac{1}{11}\left(\begin{array}{rrr}6&\,\,9&2\\-7&6&-6\\-6&2&9\end{array}\right),
\ee
and setting $\m=\frac{1}{\sqrt{22}}(3\e_1+2\e_2-3\e_3)$ we find that
\be
\label{ss4}
(-\1+2\m\otimes\m)\Q=\left(\begin{array}{rrr}0&0&-1\\1&0&0\\0&-1&0\end{array}\right)\in P^{24},
\ee
so that $\1+\tilde\av\otimes\tilde\n$ is a twin by \eqref{twina}. It is not a Type 1 or Type 2 twin because by computation neither $(-\1+2\tilde\n\otimes\tilde\n)\Q$ nor $\displaystyle\left(-\1+2\frac{\tilde\av\otimes\tilde\av}{|\tilde\av|^2}\right)\Q$ belong to $P^{24}$.\qed
\end{proof}

Next we look at the conjugates of the slips for bcc mentioned in \eqref{gh2}-\eqref{gh4}.
\begin{theorem}
\label{bccslipconj1}\;\\
$\rm(i)$ The conjugates of the slips for bcc with
\be
\label{ss5}
\av\otimes\n\in\{\pm \e_k\otimes(\e_i+\e_j), \e_k\otimes(\e_i-\e_j), i,j,k\text{ distinct}\}
\ee
are Type $2$ twins.\\
$\rm(ii)$ The conjugates of the slips \eqref{gh3}, \eqref{gh4} for bcc  with $(112)$ and $(123)$ slip planes are  twins that are neither Type $1$ nor Type $2$.\\
\end{theorem}
\begin{proof}
(i) As in Remark \ref{Prem} it suffices to consider the case $\av\otimes\n=\e_3\otimes(\e_1+\e_2)$, for   which
\be
\label{ss6}
\tilde\av\otimes\tilde\n=\frac{1}{3}(\e_1+\e_2-\e_3)\otimes(2\e_3+\e_1+\e_2)
\ee
and 
\be
\label{ss7}
\Q=(\1+\tilde\av\otimes\tilde\n)(\1+\av\otimes\n)^{-1}=\frac{1}{3}\left(\begin{array}{rrr}2&-1&\,\,2\\-1&2&2\\-2&-2&1\end{array}\right).
\ee
Then  $\1+\tilde\av\otimes\tilde\n$ is not a slip by Theorem \ref{slipslip}, and
\be
\label{ss8}
\left(-\1+2\frac{\tilde\av\otimes\tilde\av}{|\tilde\av|^2}\right)\Q=\left(\begin{array}{rrr}0&\,\,1&0\\1&0&0\\0&0&-1\end{array}\right)\in P^{24},
\ee
so that $\1+\tilde\av\otimes\tilde\n$ is a Type 2 twin. \\

(ii) For the $(112)$ slips it suffices to consider the case 
\be
\label{ss9}
\av\otimes\n=\half(\e_1+\e_2-\e_3)\otimes(\e_1+\e_2+2\e_3),
\ee
for which
\be
\label{ss10}
\tilde\av\otimes\tilde\n=\frac{1}{34}(-\e_1-\e_2+7\e_3)\otimes (7\e_1+7\e_2+2\e_3)
\ee
and
\be
\label{ss11}
\Q=(\1+\tilde\av\otimes\tilde\n)(\1+\av\otimes\n)^{-1}=\frac{1}{17}\left(\begin{array}{rrr}8&-9&-12\\-9&8&-12\\12&12&-1\end{array}\right).
\ee
Then  $\1+\tilde\av\otimes\tilde\n$ is not a slip by Theorem \ref{slipslip}, and to show that it is a twin we let $\m=\frac{1}{34}(3\e_1+3\e_2-4\e_3)$, so that
\be
\label{ss12}
(\-1+2\m\otimes\m)\Q=-\1+2\e_3\otimes\e_3\in P^{24}.
\ee
It is not a Type 1 or Type 2 twin because by computation neither $(-\1+2\tilde\n\otimes\tilde\n)\Q$ nor $\displaystyle\left(-\1+2\frac{\tilde\av\otimes\tilde\av}{|\tilde\av|^2}\right)\Q$ belong to $P^{24}$.\\

For the  $(123)$ slips it suffices to consider the case
\be
\label{ss13}
\av\otimes\n=\half(\e_1+\e_2-\e_3)\otimes(\e_1+2\e_2+3\e_3),
\ee
for which
\be
\label{ss14}
\tilde\av\otimes\tilde\n=\frac{1}{58}(-5\e_1-3\e_2+13\e_3)\otimes(7\e_1+10\e_2+5\e_3)
\ee
and
\be
\label{ss15}
\Q=(\1+\tilde\av\otimes\tilde\n)(\1+\av\otimes\n)^{-1}=\frac{1}{29}\left(\begin{array}{rrr}12&-24&-11\\-16&3&-24\\21&16&-12\end{array}\right).
\ee
Letting $\m=\frac{1}{\sqrt{29}}(3\e_1-4\e_2-2\e_3)$ we find that
\be
\label{ss16}
(-\1+2\m\otimes\m)\Q=\left(\begin{array}{rrr}0&\,\,0&\,\,1\\0&1&0\\-1&0&0\end{array}\right),
\ee
so that $\1+\tilde\av\otimes\tilde\n$ is a twin, and as for the $(112)$ case one can check that it is not of Type 1 or Type 2.\qed
\end{proof}
\begin{remark}
It is not in general true that the conjugate of a slip is a twin or a slip. For example we can take 
\be
\label{ss17}
\av\otimes\n=2(\e_1+3\e_2+5\e_3)\otimes(\e_1-2\e_2+\e_3),
\ee
which is a slip for both bcc and simple cubic, albeit for a high value of $|\av|^2$. A calculation then shows that
\be
\label{ss18}
\Q=(\1+\tilde\av\otimes\tilde\n)(\1+\av\otimes\n)^{-1}=\frac{1}{211}\left(\begin{array}{rrr}129&114&-122\\94&-177&-66\\-138&-14&-159\end{array}\right),
\ee
and $\Q\Pv$ cannot equal $-\1+2\m\otimes\m$ for any $\m$ because the  absolute values of the entries of $\Q$ are all different, and the action of multiplying on the right by $\Pv$ permutes the columns of $\Q$ with possible changes of sign.

This shows in particular that there are rank-one connections which are neither twins nor slips. Such rank-one connections without mirror symmetry across the interface are observed for martensitic phase transformations, e.g. for ${\rm LaNbO}_4$ (Jian \& Wayman \cite{jianwayman1995}, Jian \& James \cite{jianjames1997}) and for NiMnGa (Seiner et al. \cite{seiner2019}).  

Conjugates of slips are observed in `kink deformations'; see Inamura \cite{inamura2019}, and for earlier work Starkey \cite{starkey1968}.
\end{remark}

\begin{acknowledgements}
This paper grew out of an invitation to give the 2019 Le\c{c}ons Jacques-Louis Lions at Sorbonne Universit\'{e} Paris, and was developed during visits to the Hong Kong Institute for Advanced Study and the Hausdorff Institute of Mathematics (funded by the Deutsche Forschungsgemeinschaft  under Germany's Excellence Strategy – EXC-2047/1 – 390685813). The research forms part of the project {\it Mathematical theory of polycrystalline materials} supported by the EPSRC grant EP/V00204X.
I am grateful to Sergio Conti, Duvan Henao, Tomonari Inamura, Dick James, Kostas Koumatos,  Michael Ortiz, Hanu\v{s} Seiner and Dave Srolowitz for their interest and helpful suggestions, and to the referees for careful reading which led to improvements in presentation.
\end{acknowledgements}

%
 \section*{Conflict of interest}

The author declares that he has no conflict of interest.


\begin{thebibliography}{10}
\providecommand{\url}[1]{{#1}}
\providecommand{\urlprefix}{URL }
\expandafter\ifx\csname urlstyle\endcsname\relax
  \providecommand{\doi}[1]{DOI~\discretionary{}{}{}#1}\else
  \providecommand{\doi}{DOI~\discretionary{}{}{}\begingroup
  \urlstyle{rm}\Url}\fi

\bibitem{ashcroftmermin}
Ashcroft, N.W., Mermin, N.D.: Solid state physics.
\newblock New York: Holt, Rinehart and Winston (1976)

\bibitem{j56}
Ball, J.M.: Mathematical models of martensitic microstructure.
\newblock Materials Science and Engineering A \textbf{78}(1-2), 61--69 (2004)

\bibitem{p37}
Ball, J.M., Carstensen, C.: Interaction of martensitic microstructures in
  adjacent grains.
\newblock In: Proceedings of International Conference on Martensitic
  Transformations, Chicago. TMS, Springer (2017).
\newblock ArXiv:1708.03286

\bibitem{j32}
Ball, J.M., James, R.D.: Fine phase mixtures as minimizers of energy.
\newblock Arch. Ration. Mech. Anal. \textbf{100}, 13--52 (1987)

\bibitem{beviscrocker1968}
Bevis, M., Crocker, A.G.: Twinning shears in lattices.
\newblock Proceedings of the Royal Society of London. Series A. Mathematical
  and Physical Sciences \textbf{304}(1476), 123--134 (1968)

\bibitem{beviscrocker1969}
Bevis, M., Crocker, A.G.: Twinning modes in lattices.
\newblock Proceedings of the Royal Society of London. A. Mathematical and
  Physical Sciences \textbf{313}(1515), 509--529 (1969)

\bibitem{bilbycrocker1965}
Bilby, B.A., Crocker, A.G.: The theory of the crystallography of deformation
  twinning.
\newblock Proceedings of the Royal Society of London. Series A. Mathematical
  and physical sciences \textbf{288}(1413), 240--255 (1965)

\bibitem{chalmers1952slip}
Chalmers, B., Martius, U.M.: Slip planes and the energy of dislocations.
\newblock Proceedings of the Royal Society of London. Series A. Mathematical
  and Physical Sciences \textbf{213}(1113), 175--185 (1952)

\bibitem{chenmaddin}
Chen, N.K., Maddin, R.: Slip planes and the energy of dislocations in a
  bodycentered cubic structure.
\newblock Acta Metallurgica \textbf{2}(1), 49--51 (1954)

\bibitem{chenetal13}
Chen, X., Srivastava, V., Dabade, V., James, R.D.: Study of the cofactor
  conditions: Conditions of supercompatibility between phases.
\newblock Journal of the Mechanics and Physics of Solids \textbf{61}(12),
  2566--2587 (2013).
\newblock \doi{https://doi.org/10.1016/j.jmps.2013.08.004}.
\newblock
  \urlprefix\url{https://www.sciencedirect.com/science/article/pii/S002250961300149X}

\bibitem{christian2002}
Christian, J.W.: The Theory of Transformation in Metals and Alloys Pergamon.
\newblock Elsevier, New York (2002)

\bibitem{christianmahajan1995}
Christian, J.W., Mahajan, S.: Deformation twinning.
\newblock Progress in materials science \textbf{39}(1-2), 1--157 (1995)

\bibitem{cordier2002}
Cordier, P.: Dislocations and slip systems of mantle minerals.
\newblock Reviews in Mineralogy and Geochemistry \textbf{51}(1), 137--179
  (2002)

\bibitem{Ericksen77}
Ericksen, J.L.: Special topics in elastostatics.
\newblock In: C.S. Yih (ed.) Advances in Applied Mechanics, vol.~17, pp.
  189--244. Academic Press (1977)

\bibitem{forclaz99}
Forclaz, A.: A simple criterion for the existence of rank-one connections
  between martensitic wells.
\newblock J. Elasticity \textbf{57}, 281--305 (1999)

\bibitem{inamura2019}
Inamura, T.: Geometry of kink microstructure analysed by rank-1 connection.
\newblock Acta Materialia \textbf{173}, 270--280 (2019).
\newblock \doi{https://doi.org/10.1016/j.actamat.2019.05.023}.
\newblock
  \urlprefix\url{https://www.sciencedirect.com/science/article/pii/S1359645419303052}

\bibitem{jackson2012handbook}
Jackson, A.G.: Handbook of crystallography: for electron microscopists and
  others.
\newblock Springer Science \& Business Media (2012)

\bibitem{jamesprivate}
James, R.D.: Private communication

\bibitem{jaswondove1956}
Jaswon, M.A., Dove, D.B.: Twinning properties of lattice planes.
\newblock Acta Crystallographica \textbf{9}(8), 621--626 (1956)

\bibitem{jaswondove1957}
Jaswon, M.A., Dove, D.B.: The prediction of twinning modes in metal crystals.
\newblock Acta Crystallographica \textbf{10}(1), 14--18 (1957)

\bibitem{jaswondove1960}
Jaswon, M.A., Dove, D.B.: The crystallography of deformation twinning.
\newblock Acta Crystallographica \textbf{13}(3), 232--240 (1960)

\bibitem{jianjames1997}
Jian, L., James, R.: Prediction of microstructure in monoclinic {LaNbO4} by
  energy minimization.
\newblock Acta Materialia \textbf{45}(10), 4271--4281 (1997).
\newblock \doi{10.1016/S1359-6454(97)00080-3}

\bibitem{jianwayman1995}
Jian, L., Wayman, C.M.: Electron back scattering study of domain structure in
  monoclinic phase of a rare-earth orthoniobate {LaNbO4}.
\newblock Acta metallurgica et materialia \textbf{43}(10), 3893--3901 (1995)

\bibitem{Khachaturyan83}
Khachaturyan, A.G.: Theory of Structural Transformations in Solids.
\newblock John Wiley (1983)

\bibitem{kiho1954}
Kih\^{o}, H.: The crystallographic aspect of the mechanical twinning in metals.
\newblock Journal of the Physical Society of Japan \textbf{9}(5), 739--747
  (1954).
\newblock \doi{10.1143/JPSJ.9.739}.
\newblock \urlprefix\url{https://doi.org/10.1143/JPSJ.9.739}

\bibitem{kiho1958}
Kih\^{o}, H.: The crystallographic aspect of the mechanical twinning in {Ti}
  and $\alpha$-{U}.
\newblock Journal of the Physical Society of Japan \textbf{13}(3), 269--272
  (1958).
\newblock \doi{10.1143/JPSJ.13.269}.
\newblock \urlprefix\url{https://doi.org/10.1143/JPSJ.13.269}

\bibitem{ortizrepetto1999}
Ortiz, M., Repetto, E.A.: Nonconvex energy minimization and dislocation
  structures in ductile single crystals.
\newblock Journal of the Mechanics and Physics of Solids \textbf{47}(2),
  397--462 (1999)

\bibitem{pitterizanzotto03}
Pitteri, M., Zanzotto, G.: Continuum models for phase transitions and twinning
  in crystals.
\newblock Chapman \& Hall/CRC (2003)

\bibitem{seiner2019}
Seiner, H., Chulist, R., Maziarz, W., Sozinov, A., Heczko, O., Straka, L.:
  Non-conventional twins in five-layer modulated {Ni-Mn-Ga} martensite.
\newblock Scripta Materialia \textbf{162}, 497--502 (2019).
\newblock \doi{https://doi.org/10.1016/j.scriptamat.2018.12.020}.
\newblock
  \urlprefix\url{https://www.sciencedirect.com/science/article/pii/S1359646218307553}

\bibitem{starkey1968}
Starkey, J.: The geometry of kink bands in crystals—a simple model.
\newblock Contributions to Mineralogy and Petrology \textbf{19}(2), 133--141
  (1968)

\bibitem{weinberger2013slip}
Weinberger, C.R., Boyce, B.L., Battaile, C.C.: Slip planes in bcc transition
  metals.
\newblock International Materials Reviews \textbf{58}(5), 296--314 (2013).
\newblock \doi{10.1179/1743280412Y.0000000015}.
\newblock \urlprefix\url{https://doi.org/10.1179/1743280412Y.0000000015}

\end{thebibliography}

\end{document}